\documentclass[%
 reprint,
 amsmath,amssymb,
 aps,
]{revtex4-2}

\usepackage{graphicx}
\usepackage{xcolor}
\usepackage[normalem]{ulem}
\usepackage[english]{babel}

\usepackage{amsmath}
\usepackage{amssymb,amsthm}
\usepackage{stmaryrd}
\usepackage{dsfont}
\usepackage{braket}

\usepackage{here}
\usepackage{array}
\usepackage{booktabs}
\usepackage{makecell}
\usepackage{longtable}
\usepackage{tabularx}
\usepackage{pict2e}
\usepackage{diagbox}
\newcolumntype{Y}{>{\centering\arraybackslash}X}

\usepackage{tikz}
\usepackage{tikz-cd}
\usepackage{quantikz}

\usepackage{ytableau}
\ytableausetup{boxsize=1em}

\usepackage{xcolor}
\usepackage[
  colorlinks=true,
  linkcolor=blue,
  citecolor=blue,
  urlcolor=blue
]{hyperref}

\makeatletter

\newcommand{\namedlabel}[2]{\begingroup
  #2%
  \def\@currentlabel{#2}%
  \phantomsection\label{#1}\endgroup}
\makeatother

\newtheorem{Def}{Definition}
\newtheorem{Thm}[Def]{Theorem}
\newtheorem{Cor}[Def]{Corollary}

\newtheorem{Lem}[Def]{Lemma}

\newtheorem{task}{Task}

\newcommand{\bigO}[1]{\mathcal{O}\!\left(#1\right)}

\newcommand{\bigOmega}[1]{\Omega\!\left(#1\right)}
\newcommand{\bigtilOmega}[1]{\tilde{\Omega}\!\left(#1\right)}
\newcommand{\bigtheta}[1]{\Theta\!\left(#1\right)}

\DeclareMathOperator{\Tr}{Tr}

\DeclareRobustCommand{\erase}{\bgroup\markoverwith{\textcolor{red}{\rule[.5ex]{2pt}{0.4pt}}}\ULon}

\begin{document}

\title{Sample-Query Interconversion of Block Encoding of Unknown Quantum States}
\author{Manaki Arihara}
\affiliation{Department of Physics, The University of Tokyo,
7-3-1 Hongo, Bunkyo-ku, Tokyo 113-0033, Japan}
\author{Mio Murao}
\affiliation{Department of Physics, The University of Tokyo,
7-3-1 Hongo, Bunkyo-ku, Tokyo 113-0033, Japan}
\date{\today}
\begin{abstract}
Block encoding embeds a matrix as a sub-block of a unitary matrix and serves as a fundamental input model for quantum algorithms based on quantum singular value transformation, enabling polynomial transformations of matrices encoded in unitary operators. Block encoding of unknown quantum states can be useful for quantum learning; however, the fundamental limits on converting between unknown quantum states and their block-encoding unitary channels remain poorly understood. In this paper, we investigate this convertibility in both directions. First, we prove that implementing an $\varepsilon$-approximate block-encoding unitary channel of an unknown quantum state requires $\Omega(1/\varepsilon)$ copies of the state, matching known upper bounds up to logarithmic factors. Second, we show that recovering a rank-$r$, $d$-dimensional quantum state $\rho$ given query access to its block-encoding unitary channel generally requires $\Omega((1/\lambda_{\max}(\rho))\sqrt{d/r})$ queries, where $\lambda_{\max}(\rho)$ is the maximum eigenvalue of $\rho$, revealing an unavoidable dependence on the dimension of the state. Our results identify inherent limitations of block encoding as a representation of unknown quantum states and reveal a separation between learning properties of a quantum state and generating the state itself. Using our techniques, we further establish lower bounds for specific state-generation tasks, including ground-state preparation and Gibbs-state preparation.
\end{abstract}

\maketitle
\section{Introduction}
\begin{figure}[t]
  \centering
  \includegraphics[width=\linewidth]{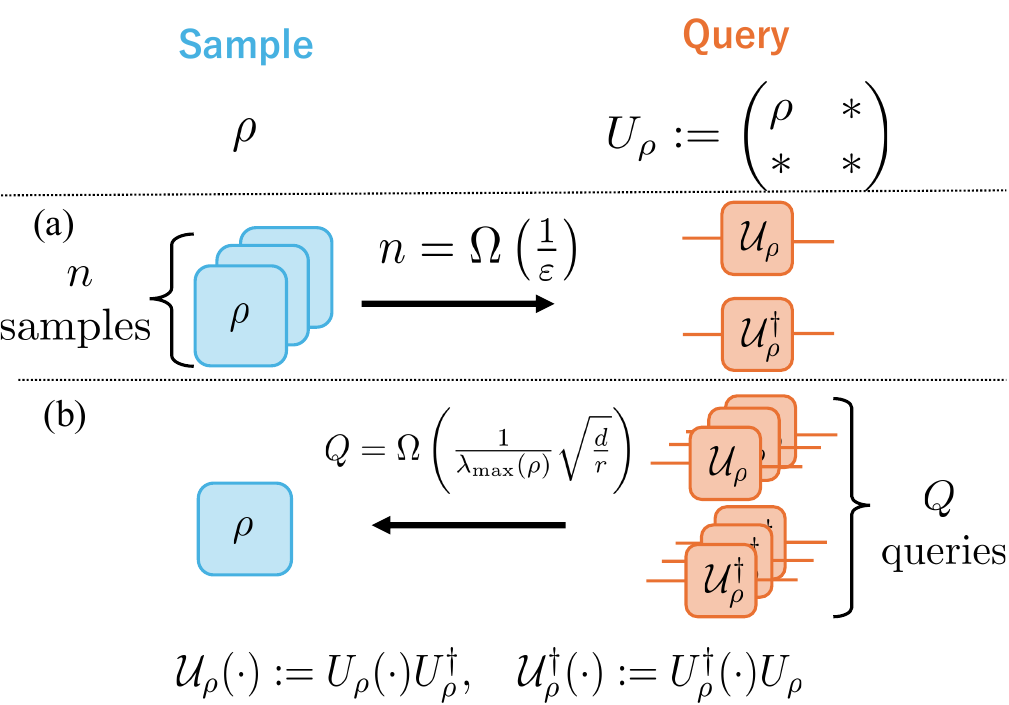}
\caption{Summary of the results. (a) Sample-to-query conversion: implementing $\varepsilon$-approximate block-encoding unitary channels $\mathcal{U}_\rho$ and $\mathcal{U}_\rho^\dagger$ given sample access to an unknown quantum state $\rho$. We show that $\bigOmega{1/\varepsilon}$ copies of $\rho$ are required for this task (Theorem~\ref{Theorem:lower bound of samples}). (b) Query-to-sample conversion: recovering an unknown quantum state $\rho$ given query access to its block-encoding unitary channels $\mathcal{U}_\rho$ and $\mathcal{U}_\rho^\dagger$. We show that $\bigOmega{(1/\lambda_{\max}(\rho))\sqrt{d/r}}$ queries are required for this task, where $\lambda_{\max}(\rho)$ is the maximum eigenvalue of $\rho$, and $d$ and $r$ are the dimension and rank of $\rho$, respectively (Theorem~\ref{theorem:lower-bound-of-queries}). This lower bound reveals an unavoidable dependence on the dimension.}

  \label{Fig:abstract}
\end{figure}

Quantum singular value transformation (QSVT) is a central subroutine in many recent quantum algorithms \cite{Gilyen2019,Low2019}, enabling polynomial transformations of the singular values of a given matrix on a quantum computer. The basic input model underlying this framework is block encoding \cite{Gilyen2019}. Block encoding embeds a target matrix $A$ as a subblock of a larger unitary matrix, allowing QSVT to perform polynomial transformations of the singular values of $A$. In many works, a block-encoding unitary channel is assumed to be given as an input. Such a block-encoding unitary channel can be efficiently implemented in \textit{white-box} settings when the target matrix $A$ is known and relevant structural information, such as its sparsity \cite{camps2023} or certain specific structures \cite{Sünderhauf2024}, is available.

On the other hand, QSVT itself is universal with respect to the embedded matrix $A$, and a full classical description of $A$ is not required to perform QSVT.
This \textit{black-box} viewpoint naturally suggests the use of QSVT for quantum learning, which aims to extract task-relevant information from unknown quantum objects, such as quantum states and channels \cite{montanaro2018}. If a block-encoding unitary channel of an unknown quantum object can be efficiently implemented, QSVT could be used to transform the embedded object and extract its desired properties. For such learning tasks, conversion between black-box quantum objects and their block-encoding unitary channels arises naturally, and understanding the resources required for this conversion is crucial for evaluating the efficiency of quantum learning algorithms.

In previous works, algorithms have been proposed that use $O(1/\varepsilon)$ copies of an unknown quantum state to implement an $\varepsilon$-approximate block-encoding unitary channel of its density matrix and estimate properties of the state \cite{Wang2024,Wang2025b}.
It has also been shown that implementing a block-encoding unitary channel of an unknown quantum channel generally requires a number of queries that depends on the dimension of the system~\cite{niwa2025}.  These results suggest that converting an unknown quantum object into its block-encoding unitary channel can incur an intrinsic resource cost. In this work, we focus on unknown quantum states and investigate the conversion between sample access to a quantum state and query access to its block-encoding unitary channel in both directions (Fig.~\ref{Fig:abstract}).

First, we consider the direction from sample access to an unknown quantum state to query access to its block-encoding unitary channel and its inverse. 
Our first main result shows that implementing quantum channels that simulate the block-encoding unitary channel and its inverse with error $\varepsilon$ in diamond norm requires $\Omega(1/\varepsilon)$ copies of the input state in general.
Thus, in the direction from samples to block-encoding unitary channels, accurate conversion to block-encoding access has an unavoidable cost, although this cost does not depend on the dimension $d$ of the state.

Second, we turn to the opposite direction, from query access to a block-encoding unitary channel and its inverse to a quantum state. We refer to this conversion task as state recovery. Our second main result establishes a dimension-dependent lower bound on the query complexity of state recovery, which holds even under a rank-$r$ promise and can be exponential in the number of qubits in the worst case. 
In contrast to the sample-to-query direction, this result reveals an unavoidable dimension-dependent cost for state recovery.
We also provide optimal algorithms for state recovery when the quantum state has a nearly flat spectrum, including projector states and pure states, thereby clarifying when recovery with tight bounds is possible under additional constraints.

Further, we show that the same lower-bound framework yields consequences for two standard state-generation tasks: ground-state preparation and Gibbs-state preparation. 
For ground-state preparation, QSVT-based eigenstate filtering has been shown to achieve near-optimal dependence on the initial overlap $\gamma$ and the spectral gap $\Delta$ \cite{Lin2020}. We construct a family of diagonal marked Hamiltonians that establishes a matching lower bound of $\Omega\left(\frac{1}{\gamma\Delta}\right)$ on the \textit{joint} dependence on $\gamma$ and $\Delta$.
For Gibbs-state preparation, previous block-encoding algorithms and sample-to-query lower bounds establish the relevance of both the inverse temperature $\beta$ and the system dimension $d$ \cite{Gilyen2019,Wang2025a}.  We construct a family of low-temperature instances that yields a query lower bound of 
$\widetilde{\Omega}\left(\beta\sqrt{d}\right)$,
matching the known worst-case scaling up to logarithmic factors when
$\beta=\Omega\left(\log d\right)$.

Taken together, our results indicate that unknown quantum states and the block-encoding unitary channels of their density matrices are not freely interchangeable resources.
Sample access to an unknown quantum state does not automatically provide arbitrarily accurate access to its block-encoding unitary channel and its inverse without incurring an $\Omega(1/\varepsilon)$ cost.
Conversely, query access to a block-encoding unitary channel of a density matrix does not generally allow one to efficiently generate the corresponding quantum state.
Thus, our results identify fundamental limitations on the convertibility between unknown quantum states and their block-encoding unitary channels and reveal a computational separation between learning properties of quantum states and generating the states themselves.

\section{Preliminaries}
\subsection{Basic Notation}

Let $\mathcal{H}$ be a $d$-dimensional Hilbert space with computational basis $\{\ket{i}\}_{i=1}^d$.
We define the normalized maximally entangled state on $\mathcal{H}\otimes\mathcal{H}$ by
\begin{equation}
  \ket{\Phi^+}=\frac{1}{\sqrt{d}}\sum_{i=1}^d \ket{i}\ket{i}.
\end{equation}

For a vector $\ket{\psi}$, we denote its Euclidean norm by
\begin{equation}
  \|\ket{\psi}\|_2=\sqrt{\braket{\psi|\psi}}.
\end{equation}
For a linear operator $A$, we denote its trace norm and operator norm by
\begin{align}
  \|A\|_1&=\Tr\sqrt{A^\dagger A},
  \\\|A\|_\infty&=\sup_{\|\ket{\psi}\|_2=1}\|A\ket{\psi}\|_2,
\end{align}
respectively, where $\|A\|_\infty$ equals the largest singular value of $A$.
For two quantum states $\rho$ and $\sigma$, their trace distance is defined as
\begin{equation}
  D_{\mathrm{tr}}(\rho,\sigma)=\frac{1}{2}\|\rho-\sigma\|_1.
\end{equation}

For two probability distributions $P$ and $Q$ on a finite set $\mathcal{X}$, the total variation distance is defined as
\begin{equation}
  d_{\mathrm{TV}}(P,Q)=\frac{1}{2}\sum_{x\in\mathcal{X}} |P(x)-Q(x)|.
\end{equation}
The Kullback--Leibler divergence \cite{Kullback1951} from $P$ to $Q$ is defined as
\begin{equation}
  D_{\mathrm{KL}}(P\|Q)=\sum_{x\in\mathcal{X}}P(x)\log\frac{P(x)}{Q(x)}.
\end{equation}
We use the convention that $0\log(0/q)=0$, and $D_{\mathrm{KL}}(P\|Q)=\infty$ if there exists $x\in\mathcal{X}$ such that $P(x)>0$ and $Q(x)=0$.
Unless otherwise stated, $\log$ denotes the natural logarithm.

\subsection{Block Encoding}
We use the standard definition of block encoding introduced in \cite{Gilyen2019}.

\begin{Def}[block encoding { \cite{Gilyen2019}}]
Let $A$ be an operator on a $d$-dimensional Hilbert space.
Let $a\in\mathbb{Z^+}$ and $\alpha\in\mathbb{R^+}$.
A unitary matrix $U$ acting on $\lceil a+\log_2 d\rceil$ qubits is called an $(\alpha,a,\varepsilon)$-block encoding of $A$
if it satisfies
\begin{equation}
  \left\| A - \alpha\,(\bra{0}^{\otimes a}\otimes I)\,U\,(\ket{0}^{\otimes a}\otimes I) \right\|_\infty\le \varepsilon.
\end{equation}
\label{Def:block-encoding}
\end{Def}
Observe that
\begin{align}
    \left\|A\right\|_\infty
    &\leq\left\|A-\alpha\,(\bra{0}^{\otimes a}\otimes I)\,U\,(\ket{0}^{\otimes a}\otimes I)\right\|_\infty+\alpha\left\|U\right\|_\infty\nonumber\\ &\leq\alpha+\varepsilon.
\end{align}
Therefore, $\alpha$ must satisfy
\begin{equation}
    \label{equation:condition_of_BE}
    \alpha\geq\left\|A\right\|_\infty-\varepsilon
\end{equation}
\section{Sample Lower Bound for Implementing a block-encoding unitary channel}
We consider a task that converts sample access to an unknown quantum state into query access to its block-encoding unitary channel and its inverse. 
We formally state the task as follows.
\begin{task}
    \label{task:perform block-encoding unitary channel}
    Let $\alpha\ge1,a\,(\in\mathbb{N})>2,\varepsilon\in(0,1/4)$ and $\rho$ be an unknown quantum state on the $d$-dimensional Hilbert space.
    Given sample access to $n$ copies of $\rho$, the goal is to implement a single use of each of the
quantum channels $\mathcal{E}_\rho$ and $\mathcal{E}_\rho^\dagger$ such that
    \begin{equation}
        \left\|\mathcal{E}_\rho-\mathcal{U}_\rho\right\|_\diamond\le\varepsilon, \quad\left\|\mathcal{E}_\rho^\dagger-\mathcal{U}_\rho^\dagger\right\|_\diamond\le\varepsilon,
        \label{equation:desired quantum channel}
    \end{equation}
    where $\mathcal{U}_\rho(\cdot):=U_\rho(\cdot) U_\rho^\dagger$ and $\mathcal{U}_\rho^\dagger(\cdot):=U_\rho^\dagger(\cdot) U_\rho$ denote the unitary channels induced by $U_\rho$ and $U_\rho^\dagger$, respectively.
    We refer to $U_\rho$ as an $(\alpha,a,0)$-block encoding of $\rho$.
\end{task}
Note that $\mathcal{E}_\rho^\dagger$ does not denote the adjoint of $\mathcal{E}_\rho$; rather it denotes a quantum channel approximating $\mathcal{U}_\rho^\dagger$.

For constant $\alpha$ and $a$, the quantum channels in Task \ref{task:perform block-encoding unitary channel} can be implemented using the following known algorithm.

\begin{Lem}[ \cite{Wang2025a} Lemma 2.21]
    For every $\varepsilon\in(0,1)$ and given sample access to copies of a quantum state $\rho$, 
    we can implement a quantum channel $\mathcal{E}$ using
    \begin{equation}
        O\!\left(\frac{1}{\varepsilon}\log^2\frac{1}{\varepsilon}\right)
    \end{equation}
    copies of $\rho$ such that $\left\|\mathcal{E} - \mathcal{U}_\rho\right\|_\diamond\le\varepsilon$ and $U_\rho$ is a $(2, 4, 0)$-block-encoding of $\rho$.
    Moreover, we can also implement a quantum channel $\mathcal{E}^\dagger$ with the same sample complexity such that 
    $\left\|\mathcal{E}^\dagger - \mathcal{U}^\dagger_\rho\right\|_\diamond\le\varepsilon$.
\end{Lem}
Roughly speaking, this algorithm consists of two steps.
First, we use density matrix exponentiation \cite{Seth2014,Kimmel2017} to implement a quantum channel that approximates a unitary channel $e^{-i \rho t}(\cdot) e^{i\rho t}$.
Second, we apply QSVT to approximate the logarithm of the unitary operator.
However, it is not clear whether the resulting sample complexity reflects an inherent limitation of the access model or merely an artifact of the known algorithms.

We establish a lower bound for this task that is tight up to logarithmic factors.
\begin{Thm}
    Any algorithm that implements a single use of each of the quantum channels $\mathcal{E}_\rho$ and $\mathcal{E}_\rho^\dagger$ defined in Task \ref{task:perform block-encoding unitary channel} requires
\begin{equation}
    \bigOmega{\frac{1}{\varepsilon}}
\end{equation}
copies of $\rho$ in the worst case in $\alpha=\Theta(1)$ regime.
\label{Theorem:lower bound of samples}
\end{Thm}
\begin{proof}[Proof sketch]
    This lower bound can be shown via reduction from a quantum state discrimination task.
    Consider two quantum states
    \begin{equation}
    \rho_+:=
    \begin{pmatrix}
        \frac{1}{2}+\varepsilon & 0 \\
        0 & \frac{1}{2}-\varepsilon
    \end{pmatrix},
    \quad
    \rho_-:=
    \begin{pmatrix}
        \frac{1}{2}-\varepsilon & 0 \\
        0 & \frac{1}{2}+\varepsilon
    \end{pmatrix}
    \label{equation:discriminated quantum state}
    \end{equation}
to be discriminated, whose matrix elements are represented in the computational basis.
Assume that we have access to $\mathcal{E}_\rho$ and $\mathcal{E}_\rho^\dagger$, where $\rho$ is promised to be $\rho_+$ or $\rho_-$.

Observe that $U_{\rho_\pm}$, which is an $(\alpha,a,0)$-block encoding of
$\rho_\pm$, can also be regarded as an $(\alpha,a+1,0)$-block encoding
of the scalar $1/2\pm\delta$.
We apply QSVT using a polynomial approximation to the function
\begin{equation}
    f(x)
    :=
    \frac{1}{4}
    \left(
        \operatorname{sign}\left(x+\frac{1}{2\alpha}\right)
        -
        \operatorname{sign}\left(x-\frac{1}{2\alpha}\right)
    \right).
\end{equation}
Since the encoded values $(1/2\pm\varepsilon)/\alpha$ lie on opposite
sides of the threshold $1/(2\alpha)$, QSVT allows us to distinguish
$\rho_+$ from $\rho_-$ using $O(1/\varepsilon)$ queries to $\mathcal{E}_\rho$ and $\mathcal{E}_\rho^\dagger$ for $\alpha=\Theta(1)$.

Suppose that each of the quantum channels $\mathcal{E}_\rho$ and $\mathcal{E}_\rho^\dagger$ can be implemented using $n$ copies of $\rho$.  Then, we can distinguish  whether $\rho$ is $\rho_+$ or $\rho_-$ with constant probability using $\bigO{n/\varepsilon}$ copies of $\rho$.
On the other hand, distinguishing $\rho_+$ and $\rho_-$ with constant success probability requires $\bigOmega{1/\varepsilon^2}$ copies.
Therefore, $n=\bigOmega{1/\varepsilon}$ holds. The full proof is given in Appendix~\ref{appendix:state-to-query}.
\end{proof}

We emphasize the difference between our result and that of Ref.~\cite{Wang2025a}.
In their work, the cost of replacing a $Q$-query algorithm accessing a
block-encoding unitary channel with a sample-based algorithm is characterized.
Their lower bounds apply in a multi-query regime and show that, for suitable
circuit families, simulating the entire computation requires
$\bigOmega{Q^2/\varepsilon}$ copies of $\rho$.
In contrast, our result focuses on the elementary one-shot conversion underlying
such algorithms, namely, implementing a single use of a block-encoding unitary
channel and its inverse.
This one-shot regime is not captured by the lower bound of
Ref.~\cite{Wang2025a}, whose proof requires $Q$ to be sufficiently large as a
function of the target accuracy.
We show that even this primitive sample-to-query conversion has an intrinsic cost.

\section{Query lower bound for the state recovery task}

Next, we consider the opposite conversion task: recovering an unknown quantum state from query access to its block-encoding unitary channel and its inverse.
We refer to this task as the {\em state recovery task}. The task is formally defined as follows.

\begin{task}[State Recovery Task]
Let $\varepsilon\in(0,1)$.
Suppose that $U_\rho$ is an $(\alpha,a,0)$-block encoding of $\rho$, where $\rho$ is an unknown quantum state on the $d$-dimensional Hilbert space.
Given query access to the
unitary channels 
$\mathcal{U}_\rho$ and $\mathcal{U}_\rho^\dagger$, induced by $U_\rho$ and $U_\rho^\dagger$, respectively, the goal is to output (recover) a single copy of a quantum state $\hat{\rho}$ such that
\begin{equation}
    \left\|\hat{\rho}-\rho\right\|_1\le \varepsilon.
\end{equation}
\label{task:state-recovery-task}
\end{task}
Note that the recovered quantum state may remain unknown in the sense that no classical description of the state is required.

To characterize the most general strategy for this task within the quantum circuit model, we use the quantum comb formalism \cite{Chiribella2008}, which describes general higher-order transformations of quantum channels. In our setting, the relevant transformation is a supermap whose input consists of $Q$ query slots, each of which can be filled with either $\mathcal U_\rho$ (forward access) or $\mathcal U_\rho^\dagger$ (inverse access), and whose output is a quantum state.  Such a channel-to-state supermap is included in the quantum comb formalism by regarding a quantum state as a quantum channel with a trivial input system. Accordingly, the most general strategy can be represented by the following quantum circuit:
\begin{widetext}
    \begin{equation}
        \Tr_{EA} \left[\left(\mathcal{C}_Q\circ\mathcal{U}_\rho^{(s_Q)}\circ\mathcal{C}_{Q-1}\circ\mathcal{U}_\rho^{(s_{Q-1})}\cdots\mathcal{C}_1\circ\mathcal{U}_\rho^{(s_1)}\circ\mathcal{C}_0\right)\,[\ket{0}\!\bra{0}_E\otimes \ket{0}\!\bra{0}^{\otimes a}_A\otimes\proj{\mathrm{init}}]\,\right],
    \end{equation}
\end{widetext}
where $s_i\in\{\pm\}$ specifies whether the $i$th query uses the forward or inverse block-encoding unitary channel, namely, $\mathcal{U}_\rho^{(+)} =\mathcal{U}_\rho$ (forward access) and $\mathcal{U}_\rho^{(-)} =\mathcal{U}_\rho^\dagger$ (inverse access). The unitary channels $\{\mathcal{C}_i\}_{i=0}^Q$ represent arbitrary interleaving operations independent of $\rho$. By enlarging the environment system $E$ if necessary, these interleaving operations can be taken to be unitary without loss of generality, since any quantum channel admits a unitary dilation on an extended system by Stinespring dilation. As illustrated in Fig.~\ref{fig:circuit2}, each query channel $\mathcal{U}_\rho^{(s_i)}$ acts jointly on the auxiliary system $A$ (the second wire) and the target system (the third wire), on which the recovered state $\hat{\rho}$ is output, while each $\mathcal{C}_i$ acts jointly on these systems and the environment system $E$ (the first wire).

\begin{figure}[h]
    \centering
    \includegraphics[width=\linewidth]{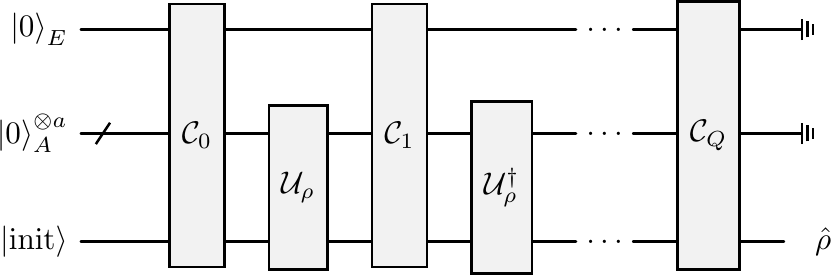}
    \caption{The quantum circuit representing the most general strategy for Task~\ref{task:state-recovery-task}. The first wire represents an environment system $E$, the second wire represents an auxiliary system $A$ and the third wire represents the target system on which the recovered state $\hat{\rho}$ is output.} 
    \label{fig:circuit2}

\end{figure}

\subsection{Query lower bound}
We now establish a lower bound on the number of queries required for state recovery from block-encoding access.
We show that the query complexity exhibits an unavoidable dependence on the dimension $d$ with the lower bound characterized in terms of the rank $r$ and the largest eigenvalue $\lambda_{\max}(\rho)$ of the unknown state $\rho$.

\begin{Thm}[Lower bound]\label{theorem:lower-bound-of-queries}
    Let $1\le r\le d/2$, and let $U_\rho$ be an $(\alpha,a,0)$-block encoding of an unknown rank-$r$ quantum state $\rho$.
    Let $\mathcal{U}_\rho$ and $\mathcal{U}_\rho^\dagger$ denote the unitary channels induced by $U_\rho$ and $U_\rho^\dagger$, respectively, and let $\lambda_{\max}(\rho)$ denote the largest eigenvalue of $\rho$.
    Any algorithm that succeeds in state recovery with trace-norm error at most $\varepsilon$ requires
    \begin{equation}
        \Omega\left(\frac{\alpha}{\lambda_{\max}(\rho)}\sqrt{\frac{d}{r}}(1-\varepsilon)\right)
    \end{equation}
    queries in total to $\mathcal{U}_\rho$ and $\mathcal{U}_\rho^\dagger$ in the worst case.
\end{Thm}

\begin{proof}[Proof sketch]
We show this lower bound by reduction from a quantum search problem \cite{Grover1996, Zalka1999, Boyer1999, Bennett1997}.
Let $S \subset \{1, 2, \cdots, d \}$ with $|S|=r$, and consider the function $f$
\begin{equation}
    f(x):=
        \begin{cases}
            \arcsin\frac{\lambda_x}{\alpha} \quad&\left(x\in S\right),\\
            0\quad&\left(x\notin S\right),
        \end{cases}
\end{equation}
where $\lambda_x >0$ and $\sum_{x\in S}\lambda_x=1$ holds.
The search task is to find $z\in S$ with constant success probability.
We introduce a phase oracle
\begin{equation}
    O_f\ket{x}=e^{if(x)}\ket{x}.
\end{equation}
Define 
\begin{equation}
   \rho_S:=\sum_{x \in S} \lambda_x \ket{x}\bra{x}.
\end{equation}
Let $U_f$ be an $(\alpha,a,0)$-block-encoding unitary of $\rho_S$.
It has been shown in Ref.~\cite{Gilyen2019} (see also Lemma~\ref{lem:block-encoding} in Appendix~\ref{appendix:query-to-sample}) that, using $O(1)$ queries to the controlled phase oracle 
\begin{equation}
 c\text{-}O_f=\ket{0}\bra{0} \otimes \mathbb{I} + \ket{1}\bra{1} \otimes O_f   
\end{equation} 
and its inverse $c\text{-}O_f^\dagger$, we can implement the block-encoding unitary channels $\mathcal{U}_f$ and $\mathcal{U}_f^\dagger$ induced by $U_f$ and $U_f^\dagger$, respectively.

Therefore, each query to the corresponding block-encoding unitary channel $\mathcal{U}_f$ or its inverse $\mathcal{U}_f^\dagger$ can be implemented using $O(1)$ queries to $c\text{-}O_f$ and $c\text{-}O_f^\dagger$. 

Suppose that a state-recovery algorithm outputs a state $\hat{\rho}$
satisfying
\begin{equation}
    \|\hat{\rho}-\rho_S\|_1\le\varepsilon
\end{equation}
using $Q$ queries to $\mathcal{U}_f$ and $\mathcal{U}_f^\dagger$.
Measuring $\hat{\rho}$ in the computational basis then yields an element
of $S$ with probability at least $1-\varepsilon/2$.
Therefore, the state-recovery algorithm gives a quantum search algorithm
using $O(Q)$ queries to $c\text{-}O_f$ and $c\text{-}O_f^\dagger$.

Then the corresponding search problem requires
\begin{equation}
    \Omega\left(
        \frac{\alpha}{\lambda_{\max}(\rho_S)}
        \sqrt{\frac{d}{r}}
        (1-\varepsilon)
    \right)
\end{equation}
queries.
Consequently, the same lower bound holds for the query complexity of
state recovery.
The full proof is given in Appendix~\ref{appendix:query-to-sample}.
\end{proof}

Even in the pure-state case, where $r=1$ and $\lambda_{\max}(\rho)=1$, the lower bound becomes $\bigOmega{\alpha\sqrt{d}}$, showing that a dimension dependence already arises.
More generally, for a rank-$r$ state in the black-box setting, we cannot assume any particular value of $\lambda_{\max}(\rho)$.
Since
\begin{equation}
    \frac{1}{r}\le\lambda_{\max}(\rho)\le1,
\end{equation}
we have $1\le1/\lambda_{\max}(\rho)\le r$.
Hence, in the worst case, the lower bound becomes $\bigOmega{\alpha\sqrt{rd}}$.
Therefore, the dimension dependence is unavoidable.
This worst case corresponds to states with a flat spectrum, for which we provide an optimal strategy for state recovery in the next subsection.

\subsection{Optimal recovery}
\begin{Cor}
    If $\rho$ is a rank-$r$ projector state, then there exists an algorithm that outputs an $\varepsilon$-approximation $\hat{\rho}$ to $\rho$ in trace norm with constant success probability using
    \begin{equation}
        \bigO{\alpha\sqrt{rd}}
    \end{equation}
    queries to $\mathcal{U}_\rho$ and $\mathcal{U}_\rho^\dagger$.
    \label{cor:recover-projection}
\end{Cor}

\begin{proof}
Let $\Pi_r$ be the projector onto an $r$-dimensional subspace of the Hilbert space.
We construct a normalized canonical purification of $\rho=\frac{\Pi_r}{r}$, namely,
\begin{equation}
     \ket{\Psi_r}:=\left(\sqrt{\frac{d}{r}}\Pi_r\otimes I\right)\ket{\Phi^+}.
\end{equation}
Since $\rho=\Pi_r/r$, $U_{\Pi_r/r}$ is an $(\alpha,a,0)$-block encoding of $\rho$.
Thus, we obtain
\begin{equation}
    \left(U_{\Pi_r/r}\otimes I\right)
    \ket{0}^{\otimes a}\ket{\Phi^+}
    =
    \frac{1}{\alpha\sqrt{rd}}
    \ket{0}^{\otimes a}\ket{\Psi_r}
    +
    \ket{\psi^\perp},
\end{equation}
where $\ket{\psi^\perp}$ satisfies
\begin{equation}
    \left(
    \ket{0}\!\bra{0}^{\otimes a}\otimes I
    \right)
    \ket{\psi^\perp}
    =0.
\end{equation}

Then we perform amplitude amplification \cite{Brassard2000}.
Let
\begin{equation}
    \sin\theta:=\frac{1}{\alpha\sqrt{rd}}.
\end{equation}
Amplitude amplification increases the success probability of projecting the auxiliary system onto $\ket{0}^{\otimes a}$ to a constant using 
\begin{equation}
    Q=\bigO{\frac{1}{\sin\theta}}=\bigO{\alpha\sqrt{rd}}
\end{equation}
queries in total to $\mathcal{U}_{\Pi_r/r}$ and $\mathcal{U}_{\Pi_r/r}^\dagger$.
Conditioned on obtaining the outcome $\ket{0}^{\otimes a}$ upon measuring the auxiliary system, the remaining system is in the normalized state
$\ket{\Psi_r}$. 
By tracing out the second system of $\ket{\Psi_r}\!\bra{\Psi_r}$, we obtain $\Pi_r/r$ exactly. Hence, the trace-norm error is zero and, in particular, is at most $\varepsilon$.

On the other hand, the maximum eigenvalue of $\Pi_r/r$ is $1/r$.
Thus, Theorem~\ref{theorem:lower-bound-of-queries} gives the lower bound
\begin{equation}
\bigOmega{\alpha r\sqrt{\frac{d}{r}}}=\bigOmega{\alpha\sqrt{rd}},
\end{equation}
for constant error $\varepsilon$.  Therefore, the above algorithm achieves the optimal query complexity up to constant factors.
\end{proof}

The preceding corollary applies to states whose nonzero eigenvalues are all exactly equal.
We next extend this result to states whose nonzero spectrum is nearly flat, as quantified by the condition number.
\begin{Cor}
    Let $\rho$ be a $d$-dimensional density matrix of rank $r$, and let
    \begin{equation}
        \kappa=\frac{\lambda_{\max}(\rho)}{\lambda_{\min}(\rho)}
    \end{equation}
    be its condition number, where $\lambda_{\min}(\rho)$ is the smallest nonzero eigenvalue of $\rho$.
    Suppose that $\kappa\in[1,2)$ and set $\varepsilon'=\kappa-1\in[0,1)$.
    Then, for any $\varepsilon\in(0,1)$ satisfying $\varepsilon\ge\varepsilon'$,  
    there exists an algorithm that recovers $\rho$ to trace-norm error at most $\varepsilon$ with constant success probability using
    \begin{equation}
        \bigO{\frac{\alpha}{\lambda_{\max}(\rho)}\sqrt{\frac{d}{r}}}
    \end{equation}
    queries in total to $\mathcal{U}_\rho$ and $\mathcal{U}_\rho^\dagger$.
\end{Cor}

\begin{proof}
The strategy is to prepare the state
$\rho^2/\Tr\rho^2$ by applying the block-encoding unitary channel
to one half of the maximally entangled state and performing amplitude
amplification, as in Corollary~\ref{cor:recover-projection}.
We first show that this state is sufficiently close to $\rho$.

Let $\{\lambda_i\}_{i=1}^r$ denote the nonzero eigenvalues of $\rho$.
Then,
\begin{align}
    \left\|\frac{\rho^2}{\Tr\rho^2}-\rho\right\|_1
    &=
    \sum_i
    \left|
        \frac{\lambda_i^2}{\sum_j\lambda_j^2}
        -\lambda_i
    \right|
    \nonumber\\
    &=
    \sum_i
    \lambda_i
    \left|
        \frac{\lambda_i}{\sum_j\lambda_j^2}-1
    \right|.
\end{align}
Since
\begin{equation}
    \lambda_{\min}(\rho)
    \le
    \sum_j\lambda_j^2
    \le
    \lambda_{\max}(\rho),
\end{equation}
we have, for every $i$,
\begin{equation}
    \frac{1}{\kappa}-1
    \le
    \frac{\lambda_i}{\sum_j\lambda_j^2}-1
    \le
    \kappa-1.
\end{equation}
Therefore,
\begin{equation}
    \left|
        \frac{\lambda_i}{\sum_j\lambda_j^2}-1
    \right|
    \le
    \kappa-1,
\end{equation}
and hence
\begin{equation}
    \left\|
        \frac{\rho^2}{\Tr\rho^2}-\rho
    \right\|_1
    \le
    \kappa-1
    =
    \varepsilon'
    \le
    \varepsilon.
\end{equation}

We next analyze the query complexity.
The amplitude of the good ancilla subspace is
\begin{equation}
    \left\|
        \left(\frac{\rho}{\alpha}\otimes I\right)
        \ket{\Phi^+}
    \right\|_2
    =
    \frac{\sqrt{\Tr\rho^2}}{\alpha\sqrt{d}}.
\end{equation}
Thus, amplitude amplification prepares
$\rho^2/\Tr\rho^2$ with constant success probability using
\begin{equation}
    \bigO{
        \alpha\sqrt{\frac{d}{\Tr\rho^2}}
    }
\end{equation}
queries in total to $\mathcal{U}_\rho$ and
$\mathcal{U}_\rho^\dagger$.

Since
\begin{equation}
    \Tr\rho^2
    =
    \sum_{i=1}^r\lambda_i^2
    \ge
    r\lambda_{\min}(\rho)^2,
\end{equation}
we obtain
\begin{align}
    \frac{1}{\sqrt{\Tr\rho^2}}
    &\le
    \frac{1}{\sqrt{r}\lambda_{\min}(\rho)}
    =
    \frac{\kappa}{\sqrt{r}\lambda_{\max}(\rho)}
    \nonumber\\
    &\le
    \frac{1+\varepsilon}
    {\sqrt{r}\lambda_{\max}(\rho)},
\end{align}
where the last inequality follows from
$\kappa-1=\varepsilon'\le\varepsilon$.
Therefore, the query complexity is bounded by
\begin{equation}
    \bigO{
        \frac{\alpha}{\lambda_{\max}(\rho)}
        \sqrt{\frac{d}{r}}
        (1+\varepsilon)
    }
    =
    \bigO{
        \frac{\alpha}{\lambda_{\max}(\rho)}
        \sqrt{\frac{d}{r}}
    },
\end{equation}
where we used $\varepsilon\in(0,1)$ in the last equality.
\end{proof}

These results show that state recovery incurs an unavoidable dimension-dependent cost. Nevertheless, in the regimes where our state-recovery algorithms apply, they can still offer a substantial query-complexity advantage over an incoherent strategy based on full quantum process tomography \cite{Mauro2003}.  In such a strategy, one first performs full quantum process tomography of the block-encoding unitary channel to  obtain a classical description of the channel, extracts a classical description of $\rho$ from it, and then uses the description to prepare $\rho$ as a quantum state. 

To quantify this advantage, let $d=2^n$ be the dimension of the $n$-qubit system Hilbert space, and let $D=2^ad=2^{a+n}$ denote the dimension of the joint 
Hilbert space on which the $(\alpha, a,0)$-block-encoding unitary acts. To guarantee
trace-norm error at most $\varepsilon$ in the recovered state, it
suffices to reconstruct the full block-encoding unitary channel to
diamond-norm error
\begin{equation}
    \delta
    =
    O\!\left(\frac{\varepsilon}{\alpha d}\right).
\end{equation}
Here, we use the equivalence between diamond distance for unitary
channels and operator-norm distance up to a global phase, together with
the generic inequality $\|X\|_1\le d\|X\|_\infty$.

General-purpose tomography of a $D$-dimensional unitary channel to
diamond-norm error $\delta$ has optimal query complexity
$\Theta(D^2/\delta)$ \cite{Haah2023}. Therefore, a full-tomography
strategy operated at this accuracy has query complexity
\begin{equation}
    \Theta\!\left(
        \frac{\alpha D^2 d}{\varepsilon}
    \right).
\end{equation}
For constant $\alpha$ and $a$, this reduces to
$\Theta(d^3/\varepsilon)$, demonstrating the advantage of our
state-recovery algorithms over this general-purpose tomography-based
approach in the regimes where they apply.

\section{Applications to ground-state and Gibbs-state preparation}
The search-reduction technique used to prove Theorem~\ref{theorem:lower-bound-of-queries} can also be applied to other state-generation tasks. In particular, by constructing Hamiltonians whose ground or Gibbs states encode a marked computational-basis element, the same technique yields task-specific lower bounds. We present two such applications below: ground-state preparation and Gibbs-state preparation.

For ground-state preparation, we compare our result with the eigenstate-filtering framework of Ref.~\cite{Lin2020}. Their QSVT-based algorithm achieves essentially optimal dependence on the initial overlap \(\gamma\) and the spectral gap \(\Delta\). Their lower-bound theorem establishes optimality in two one-parameter regimes, namely, \(\Omega(1/\gamma)\) for constant \(\Delta\) and \(\Omega(1/\Delta)\) for constant \(\gamma\). The corollary below is complementary and establishes a stronger joint lower bound in the simultaneous small-overlap and small-gap regime. We construct a diagonal hard family for which a marked phase of size \(\Theta(\Delta/\alpha)\), together with the uniform initial state satisfying \(\gamma=1/\sqrt d\), forces
$
\Omega(\alpha\sqrt d/\Delta)=
\Omega(\alpha/(\gamma\Delta))
$
block-encoding queries. Thus, in the block-encoding query model and for \(\alpha=\Theta(1)\), the \(1/(\gamma\Delta)\) scaling is not merely an artifact of filtering algorithms.

For Gibbs-state preparation, we compare our result with previous results \cite{Gilyen2019,Wang2025a}. Existing results give a worst-case upper bound of order \(\widetilde{O}(\beta\sqrt d)\), together with lower bounds that include the separate obstructions \(\widetilde{\Omega}(\beta)\) and \(\Omega(\sqrt d)\). These separate lower bounds do not by themselves imply a product dependence. Corollary~\ref{cor:gibbs-specific-bound} shows that, in the low-temperature regime, a diagonal marked Hamiltonian enforces the product-type lower bound
$
\Omega(\beta\sqrt d/\log d)=
\widetilde{\Omega}(\beta\sqrt d).
$
Thus, the corollary establishes a worst-case lower bound for preparing the Gibbs state itself and shows that the temperature and dimension costs can arise simultaneously rather than only as separate obstructions.

\begin{task}[Ground-state preparation from a block encoding]
Let $H$ be a Hamiltonian with a unique ground state $\ket{z}$ and spectral gap at least $\Delta$, and let $U_H$ be an $(\alpha,a,0)$-block encoding of $H$.
Given query access to block-encoding unitary channels $\mathcal{U}_H$ and $\mathcal{U}_H^\dagger$ induced by $U_H$ and $U_H^\dagger$, respectively, and an initial state $\ket{\phi_0}$ satisfying $|\braket{z|\phi_0}|=\gamma$, the goal is to output a state satisfying
\begin{equation}
    \bra{z}\hat{\rho}\ket{z}\ge\frac{3}{4}
\end{equation}.
\end{task}

\begin{Cor}[Ground-state preparation]
\label{cor:ground-state-specific-bound}
Let $\alpha=\Theta(1)$ and $0<\Delta\le \alpha$.
There exists a family of Hamiltonians $H_z=-\Delta \ket{z}\!\bra{z}$ and initial states
\begin{equation}
    \ket{\phi_0}=\frac{1}{\sqrt d}\sum_{x=1}^d\ket{x},
\end{equation}
for which any algorithm solving the ground-state preparation task with constant success probability requires
\begin{equation}
    \Omega\!\left(\frac{\sqrt d}{\Delta}\right)
    =\Omega\!\left(\frac{1}{\gamma\Delta}\right)
\end{equation}
queries in total to $\mathcal{U}_H$ and $\mathcal{U}_H^\dagger$.
\end{Cor}

\begin{proof}
For a marked element $z\in[d]$, define a phase oracle by
\begin{equation}
    O_z\ket{x}=e^{if_z(x)}\ket{x},\qquad
    f_z(x)=
    \begin{cases}
    \arcsin(\Delta/\alpha) & (x=z),\\
    0 & (x\ne z).
    \end{cases}
\end{equation}
The construction of Lemma~\ref{lem:block-encoding}, together with an additional fixed sign flip, implements an $(\alpha, a,0)$-block encoding of $H_z=-\Delta\ket{z}\!\bra{z}$ using only $O(1)$ queries to the controlled phase oracle and its inverse.
The state $\ket{\phi_0}=d^{-1/2}\sum_x\ket{x}$ has overlap $\gamma=1/\sqrt d$ with the ground state $\ket{z}$.

Suppose that the ground-state preparation task could be solved using $o(1/(\gamma\Delta))=o(\sqrt d/\Delta)$ block-encoding queries.  Then measuring the output in the computational basis would identify $z$ with probability at least $3/4$.
This would solve the corresponding marked-element search problem using $o(\sqrt d/\Delta)$ queries to the controlled phase oracle and its inverse.  However, the same hybrid argument as Lemma~\ref{Lem:BVVV}, specialized to a single marked element with marked phase $\Theta(\Delta/\alpha)$, gives a lower bound
\begin{equation}
    \Omega\left(\frac{\alpha\sqrt d}{\Delta}\right)
    =
    \Omega\left(\frac{\sqrt d}{\Delta}\right),
\end{equation}
where we used $\alpha=\Theta(1)$.  This is a contradiction.
\end{proof}

\begin{task}[Gibbs-state preparation from a block encoding]
Let $H$ be a Hamiltonian and let $\beta>0$ be the inverse temperature and let $U_H$ be an $(\alpha,a,0)$-block encoding of $H$.
Given query access to block-encoding unitary channels $\mathcal{U}_H$ and $\mathcal{U}_H^\dagger$ induced by $U_H$ and $U_H^\dagger$, respectively, the goal is to output a quantum state $\hat{\rho}$ close in trace norm to the Gibbs state
\begin{equation}
    \rho_{\mathrm{Gibbs}}(H)=\frac{e^{-\beta H}}{\Tr e^{-\beta H}}.
\end{equation}
\end{task}

\begin{Cor}[Low-temperature Gibbs-state preparation]
\label{cor:gibbs-specific-bound}
Let $\alpha=\Theta(1)$.
In the low-temperature regime $\beta\ge\log d$, there exists a family of diagonal Hamiltonians for which preparing the Gibbs state to constant trace-norm error requires
\begin{equation}
    \Omega\!\left(\frac{\beta\sqrt d}{\log d}\right)
    =\bigtilOmega{\beta\sqrt d}
\end{equation}
queries in total to $\mathcal{U}_H$ and $\mathcal{U}_H^\dagger$.
\end{Cor}

\begin{proof}
Choose $\Delta=\log d$ and assume $\beta\ge\Delta$.
For each marked element $z\in \{1,2, \cdots, d\}$, consider the diagonal Hamiltonian
\begin{equation}
    H_z=-\frac{\Delta}{\beta}\ket{z}\!\bra{z}
    +\frac{\Delta}{d\beta}\sum_{x\ne z}\ket{x}\!\bra{x}.
\end{equation}
Its Gibbs state is diagonal and assigns probability
\begin{equation}
    p_z=\frac{e^{\Delta}}{e^{\Delta}+(d-1)e^{-\Delta/d}}
\end{equation}
to the marked basis vector $\ket{z}$.  This probability is bounded below by a positive constant for this choice of $\Delta$.

Write
\begin{equation}
    H_z=cI-m\proj{z},
    \qquad
    c:=\frac{\Delta}{d\beta},
    \quad
    m:=\frac{\Delta}{\beta}\left(1+\frac{1}{d}\right).
\end{equation}
Assume $\alpha\ge c+m$ and set $\alpha_m:=\alpha-c$.
By Lemma~\ref{lem:block-encoding}, an
$(\alpha_m/m,a,0)$-block encoding $U_z$ of $\proj{z}$ can be
implemented using $O(1)$ queries to a controlled phase oracle and its
inverse.  Its marked phase is
\begin{equation}
    \theta=
    \arcsin\left(\frac{m}{\alpha_m}\right)=\bigtheta{\frac{\Delta}{\alpha\beta}}.
\end{equation}
Combining $-U_z$ with the known identity operation by the standard
two-term LCU construction gives
\begin{equation}
    \frac{c}{\alpha}I-\frac{\alpha_m}{\alpha}\frac{m}{\alpha_m}\proj{z}=\frac{H_z}{\alpha}.
\end{equation}
Thus, an $(\alpha,a+1,0)$-block encoding $U_{H_z}$ of $H_z$ can be
implemented with $O(1)$ oracle queries.

Therefore, the same phase-search lower-bound argument as
Lemma~\ref{Lem:BVVV} gives
\begin{equation}
    \Omega\left(\frac{\sqrt d}{\theta}\right)=\Omega\left(\frac{\alpha\beta}{\Delta}\sqrt d\right)
\end{equation}
queries in total to the controlled phase oracle and its inverse.

Suppose that a Gibbs-state preparation algorithm produces a state within a sufficiently small constant trace-norm error of $\rho_{\mathrm{Gibbs}}(H_z)$ using $o(\beta\sqrt d/\Delta)$ queries in total to the block-encoding unitary channels $\mathcal{U}_{H_z}$ and its inverse $\mathcal{U}_{H_z}^\dagger$.
Then, measuring the output state in the computational basis would identify $z$ with constant probability.
This contradicts the search lower bound.
Substituting $\Delta=\log d$ gives the claimed $\Omega(\beta\sqrt d/\log d)$ lower bound.
\end{proof}

\section{Conclusion}
In this work, we investigated the convertibility between sample access to unknown quantum states and query access to their block-encoding unitary channels.
First, for $\alpha=\Theta(1)$, we established an $\bigOmega{1/\varepsilon}$ lower bound on the number of copies required to implement an $\varepsilon$-approximation to the unitary channel induced by an $(\alpha, a, 0)$-block encoding of an unknown state $\rho$ and to its inverse channel. For fixed $\alpha$, this lower bound matches the known upper bound up to logarithmic factors.
Second, we established an $\bigOmega{(1/\lambda_{\max}(\rho))\sqrt{d/r}}$ lower bound, for $\alpha=\bigtheta{1}$ and constant recovery error,  on the number of queries required to recover an unknown rank-$r$ quantum state $\rho$ on a $d$-dimensional Hilbert space with the maximum eigenvalue $\lambda_{\max}(\rho)$ from query access to its block-encoding unitary channel and its inverse. In particular, state recovery exhibits an unavoidable dependence on the dimension even for pure states. We also identified 
states with nearly flat spectra as a regime in which state recovery can be achieved with optimal query complexity and showed that, in this regime, our algorithms can outperform a general-purpose approach based on full unitary-channel tomography.

We further applied the same lower-bound technique to two standard
state-generation tasks: ground-state preparation and Gibbs-state
preparation.
For ground-state preparation, we showed that, for $\alpha=\Theta(1)$, the joint dependence $\bigOmega{1/(\gamma\Delta)}$ on the initial overlap $\gamma$ and the spectral gap $\Delta$ is unavoidable for a simple diagonal marked family, complementing QSVT-based eigenstate-filtering results \cite{Lin2020}.
For Gibbs-state preparation, we established, in the low-temperature regime, a worst-case lower bound $\Omega(\beta\sqrt d/\log d)=\bigtilOmega{\beta\sqrt d}$ for algorithms that output the Gibbs state itself \cite{Gilyen2019, Wang2025a}. This scaling matches the known worst-case scaling up to logarithmic
factors.  These applications indicate that the dimension-dependent obstruction revealed by the state-recovery problem is not confined to the recovery of arbitrary unknown density matrices, but also arises in standard state-generation tasks with physically meaningful outputs, such as low-energy states or thermal states of a Hamiltonian.

For general quantum states, whether the lower bound can be matched by an upper bound remains open. Ref.~\cite{utsumi2025} gives a related construction of a canonical purification based on a block encoding of $\sqrt{\rho}$.  However, the query complexity of implementing this block encoding does not generally match our lower bound. Closing this gap remains an important open problem.

\section*{Acknowledgments}
This work was supported by Japan Society for the
Promotion of Science (JSPS) KAKENHI Grant Numbers 23K21643, the MEXT Quantum Leap Flagship Program (MEXT QLEAP) JPMXS0118069605, JPMXS0120351339, JST CREST Grant Number JPMJCR25I5, JST ASPIRE Grant Number JPMJAP25A3, JST NEXUS Grant Number JPMJNX26C9, FoPM, WINGS Program, the University of Tokyo, and IBM Quantum.

\onecolumngrid

\appendix
\section{Proof of Theorem~\ref{Theorem:lower bound of samples}}
\label{appendix:state-to-query}
Let $\rho_+$ and $\rho_-$ be the qubit states to be discriminated, defined by
\begin{equation}
	\rho_+:=
	\begin{pmatrix}
		\frac{1}{2}+\delta & 0 \\
		0 & \frac{1}{2}-\delta
	\end{pmatrix},
	\quad \rho_-:=
	\begin{pmatrix}
		\frac{1}{2}-\delta & 0 \\
		0 & \frac{1}{2}+\delta
	\end{pmatrix}.
	\label{equation:discriminated quantum state}
\end{equation}
Here, $\delta\in(0,1/4)$.
We evaluate the sample complexity required to distinguish them.

\begin{Lem}
	\label{Lemma:lower bound of sample to discriminate the states}
	Suppose we are promised that a quantum state $\rho$ is either $\rho_+$ or $\rho_-$.
	To correctly determine whether $\rho$ is $\rho_+$ or $\rho_-$ with success probability at least $2/3$,
	\begin{equation}
		\bigOmega{\frac{1}{\delta^2}}
	\end{equation}
	copies of $\rho$ are required.
\end{Lem}

\begin{proof}
	The optimal strategy for maximizing the success probability of identifying $\rho$ is the Helstrom measurement \cite{Helstrom1969, Holevo1973}.
	If we are allowed to use $n_{\mathrm{disc}}$ copies of $\rho$, the success probability $p_{\mathrm{success}}$ satisfies
	\begin{equation}
		p_{\mathrm{success}} \le \frac{1}{2} +\frac{1}{4} \left\| \rho_+^{\otimes n_{\mathrm{disc}}} -\rho_-^{\otimes n_{\mathrm{disc}}} \right\|_1.
	\end{equation}
	Since $\rho_+$ and $\rho_-$ commute, their trace distance is the same as the classical total variation distance up to the conventional factor of two:
	\begin{equation}
		\left\| \rho_+^{\otimes n_{\mathrm{disc}}} -\rho_-^{\otimes n_{\mathrm{disc}}} \right\|_1 = 2d_{\mathrm{TV}} \left( p_+^{\otimes n_{\mathrm{disc}}}, p_-^{\otimes n_{\mathrm{disc}}} \right),
	\end{equation}
	where
	\begin{equation}
		p_+=\left(\frac{1}{2}+\delta,\frac{1}{2}-\delta\right), \qquad p_-=\left(\frac{1}{2}-\delta,\frac{1}{2}+\delta\right).
	\end{equation}
	Using Pinsker's inequality \cite{Brillinger1964},
	\begin{align}
		d_{\mathrm{TV}} \left( p_+^{\otimes n_{\mathrm{disc}}}, p_-^{\otimes n_{\mathrm{disc}}} \right) &\le \sqrt{ \frac{1}{2} D_{\mathrm{KL}} \left( p_+^{\otimes n_{\mathrm{disc}}} \middle\| p_-^{\otimes n_{\mathrm{disc}}} \right) } \nonumber\\
		&= \sqrt{ \frac{n_{\mathrm{disc}}}{2} D_{\mathrm{KL}}(p_+\|p_-) } \nonumber\\
		&= \sqrt{ \frac{n_{\mathrm{disc}}}{2} \,2\delta \log\frac{1+2\delta}{1-2\delta} } \nonumber\\
		&\le \sqrt{ \frac{n_{\mathrm{disc}}}{2} \,2\delta\cdot6\delta } \nonumber\\
		&= \sqrt{6n_{\mathrm{disc}}\delta^2}.
	\end{align}
	Here, we used
	\begin{align}
		\log(1+2\delta)&\le2\delta, \\
		-\log(1-2\delta)&\le4\delta \qquad \left(\delta\in(0,1/4)\right).
	\end{align}
	To achieve success probability at least $2/3$, we require
	\begin{equation}
		\frac{1}{3} \le d_{\mathrm{TV}} \left( p_+^{\otimes n_{\mathrm{disc}}}, p_-^{\otimes n_{\mathrm{disc}}} \right) \le \sqrt{6n_{\mathrm{disc}}\delta^2}.
	\end{equation}
	Therefore,
	\begin{equation}
		n_{\mathrm{disc}} = \bigOmega{\frac{1}{\delta^2}}.
	\end{equation}
\end{proof}

Next, we show how quantum channels that simulate a block-encoding unitary channel can be used to distinguish $\rho_+$ and $\rho_-$ with sufficient probability.
We note that an $(\alpha,a,0)$-block encoding of $\rho_\pm$ can also be viewed as an $(\alpha,a+1,0)$-block encoding of the scalar $1/2\pm\delta$.
Our strategy is to distinguish the singular values using a polynomial approximation to a superposition of sign functions, namely
\begin{equation}
	U_{\rho_\pm} =
	\begin{pmatrix}
		\frac{\frac{1}{2}\pm\delta}{\alpha} & \ast \\
		\ast & \ast
	\end{pmatrix}
	\quad \xrightarrow{\mathrm{QSVT}} \quad
	\begin{pmatrix}
		f\left(\frac{\frac{1}{2}\pm\delta}{\alpha}\right) & \ast \\
		\ast & \ast
	\end{pmatrix}.
\end{equation}
Here, $f$ is the function defined in \eqref{equation:g}.
We use the following polynomial approximation to the sign function.

\begin{Lem}[\cite{Gilyen2019}, Lemma 25]
	For all $\gamma>0$ and $\varepsilon'\in(0,1/2)$, there exists an efficiently computable odd polynomial $P\in\mathbb{R}[x]$ of degree
	\begin{equation}
		\deg P = \bigO{\frac{\log(1/\varepsilon')}{\gamma}}
	\end{equation}
	such that
	\begin{align}
		\|P\|_{\infty,[-2,2]}&\le1, \\
		\|P-\operatorname{sign}\|_{ \infty,\, [-2,2]\setminus(-\gamma/2,\gamma/2) } &\le\varepsilon'.
	\end{align}
	\label{Lemma:approximate sign function}
\end{Lem}

We use the convention
\begin{equation}
	\operatorname{sign}(x):=
	\begin{cases}
		1, & x>0,\\
		0, & x=0,\\
		-1, & x<0.
	\end{cases}
\end{equation}
Define
\begin{equation}
	f(x):= \frac{1}{4} \left( \operatorname{sign}\left(x+\frac{1}{2\alpha}\right) - \operatorname{sign}\left(x-\frac{1}{2\alpha}\right) \right),
	\label{equation:g}
\end{equation}
and let
\begin{equation}
	\widetilde{f}(x):= \frac{1}{4} \left( P\left(x+\frac{1}{2\alpha}\right) - P\left(x-\frac{1}{2\alpha}\right) \right),
\end{equation}
where we set $\gamma=\delta/\alpha$ in Lemma~\ref{Lemma:approximate sign function}.

\begin{Lem}
	Let
	\begin{equation}
		D_{\delta,\alpha} := [-1,1] \setminus \left( \left[-\frac{1+\delta}{2\alpha},-\frac{1-\delta}{2\alpha}\right] \cup \left[\frac{1-\delta}{2\alpha},\frac{1+\delta}{2\alpha}\right] \right).
	\end{equation}
	Then $\widetilde{f}$ satisfies the following properties:
	\begin{itemize}
		\item $\|\widetilde{f}\|_{\infty,[-1,1]}\le\frac{1}{2}$.\
        \item $\|\widetilde{f}-f\|_{\infty,D_{\delta,\alpha}}\le\frac{\varepsilon'}{2}$.
		\item $\widetilde{f}$ is an even polynomial and
		\begin{equation}
			\deg\widetilde{f} \le \deg P = \bigO{ \frac{\alpha}{\delta} \log\left(\frac{1}{\varepsilon'}\right) }.
		\end{equation}
	\end{itemize}
	\label{Lemma:approximate g}
\end{Lem}

\begin{proof}
	\begin{itemize}
		\item We have
		\begin{align}
			\|\widetilde{f}\|_{\infty,[-1,1]} &= \frac{1}{4} \left\| P\left(x+\frac{1}{2\alpha}\right) - P\left(x-\frac{1}{2\alpha}\right) \right\|_{\infty,[-1,1]} \nonumber\\
			&\le \frac{1}{4} \left( \|P\|_{\infty,[-1+1/(2\alpha),\,1+1/(2\alpha)]} + \|P\|_{\infty,[-1-1/(2\alpha),\,1-1/(2\alpha)]} \right) \nonumber\\
			&\le\frac{1}{2}.
		\end{align}

		\item Define
		\begin{equation}
			\Delta_\pm(x) := P\left(x\pm\frac{1}{2\alpha}\right) - \operatorname{sign}\left(x\pm\frac{1}{2\alpha}\right).
		\end{equation}
		Then
		\begin{align}
			\|\widetilde{f}-f\|_{\infty,D_{\delta,\alpha}} &= \frac{1}{4} \|\Delta_+-\Delta_-\|_{\infty,D_{\delta,\alpha}} \nonumber\\
			&\le \frac{1}{4} \left( \|\Delta_+\|_{\infty,D_{\delta,\alpha}} + \|\Delta_-\|_{\infty,D_{\delta,\alpha}} \right) \nonumber\\
			&\le\frac{\varepsilon'}{2}.
		\end{align}

		\item Since $P$ is odd,
		\begin{align}
			\widetilde{f}(-x) &= \frac{1}{4} \left( P\left(-x+\frac{1}{2\alpha}\right) - P\left(-x-\frac{1}{2\alpha}\right) \right) \nonumber\\
			&= \frac{1}{4} \left( -P\left(x-\frac{1}{2\alpha}\right) + P\left(x+\frac{1}{2\alpha}\right) \right) \nonumber\\
			&= \widetilde{f}(x).
		\end{align}
		Thus, $\widetilde{f}$ is even.
		Since shifts and linear combinations do not increase the degree,
		\begin{equation}
			\deg\widetilde{f} \le \deg P = \bigO{ \frac{\alpha}{\delta} \log\left(\frac{1}{\varepsilon'}\right) }.
		\end{equation}
	\end{itemize}
\end{proof}

\begin{Lem}[\cite{Gilyen2019}, Theorem 26]
Let $U_\rho$ be an $(\alpha,a+1,0)$-block encoding of a scalar
$\rho_{11}\in[0,1]$, and let $\mathcal{U}_\rho$ and
$\mathcal{U}_\rho^\dagger$ be the unitary channels induced by
$U_\rho$ and $U_\rho^\dagger$, respectively.
Assume that $\alpha=\Theta(1)$.
Then the unitary channel induced by a $(1,a+2,0)$-block-encoding
unitary of $\widetilde{f}(\rho_{11}/\alpha)$ can be implemented using
\begin{equation}
    l=\bigO{\frac{1}{\delta}
    \log\left(\frac{1}{\varepsilon'}\right)}
\end{equation}
queries in total to $\mathcal{U}_\rho$ and
$\mathcal{U}_\rho^\dagger$, together with known phase gates and
controlled reflections about $\ket{0}^{\otimes(a+1)}$.
	\label{lemma:QSVT}
\end{Lem}

\begin{proof}
	Let
	\begin{equation}
		\Pi := \ket{0}\!\bra{0}^{\otimes(a+1)}.
	\end{equation}
	Then
	\begin{equation}
		\frac{\rho_{11}}{\alpha} \ket{0}\!\bra{0}^{\otimes(a+1)} = \Pi U_\rho\Pi.
	\end{equation}
	By Corollary 18 of \cite{Gilyen2019} and Lemma~\ref{Lemma:approximate g}, there exists a real vector
	\begin{equation}
		\Phi = (\phi_1,\phi_2,\ldots,\phi_l) \in \mathbb{R}^l
	\end{equation}
	such that
	\begin{align}
		\widetilde{f}(\Pi U_\rho\Pi) &= (\bra{+}\otimes\Pi) \left( \ket{0}\!\bra{0}\otimes U_\Phi + \ket{1}\!\bra{1}\otimes U_{-\Phi} \right) (\ket{+}\otimes\Pi), \\
		U_\Phi &:= \prod_{j=1}^{l/2} \left( e^{i\phi_{2j-1}(2\Pi-I)} U_\rho^\dagger e^{i\phi_{2j}(2\Pi-I)} U_\rho \right).
		\label{equation:QSVT circuit}
	\end{align}
	Therefore,
	\begin{equation}
		(H\otimes I) \left( \ket{0}\!\bra{0}\otimes U_\Phi + \ket{1}\!\bra{1}\otimes U_{-\Phi} \right) (H\otimes I)
	\end{equation}
	is a block encoding of $\widetilde{f}(\rho_{11}/\alpha)$.
	The unknown unitaries $U_\rho$ and $U_\rho^\dagger$ are applied unconditionally in the same order and the same number of times in the two branches of the multiplexed unitary; only the known phase gates $e^{\pm i\phi_k(2\Pi-I)}$ are coherently controlled by the branch qubit. Hence, the construction does not require controlled access to either $U_\rho$ or $U_\rho^\dagger$. Consequently, the unitary channel induced by the above block-encoding unitary can be implemented using $l$ queries in total to $\mathcal{U}_\rho$ and $\mathcal{U}_\rho^\dagger$.
\end{proof}

\begin{Lem}[\cite{gilyén2022}, Corollary 12]
	For every block encoding $\widetilde{U}\in\mathbb{C}^{d\times d}$ of $\widetilde{A}\in\mathbb{C}^{\widetilde{n}\times n}$ with $d\ge4(\widetilde{n}+n)$, and for every $A\in\mathbb{C}^{\widetilde{n}\times n}$ satisfying
	\begin{equation}
		\|A-\widetilde{A}\|_\infty + \left\| \frac{A+\widetilde{A}}{2} \right\|_\infty^2 \le1,
	\end{equation}
	there exists a block encoding $U\in\mathbb{C}^{d\times d}$ of $A$ satisfying
	\begin{equation}
		\|U-\widetilde{U}\|_\infty \le \sqrt{ \frac{2}{ 1- \left\| \frac{A+\widetilde{A}}{2} \right\|_\infty^2 } } \, \|A-\widetilde{A}\|_\infty.
	\end{equation}
	\label{lemma:robustness of block-encoding}
\end{Lem}

Lemma~\ref{lemma:robustness of block-encoding} bounds the operator-norm distance between suitable block encodings in terms of the distance between the encoded matrices.
Choose $\varepsilon'=0.01$.
On the domain $D_{\delta,\alpha}$, the intended estimate is
\begin{align}
	\|f-\widetilde{f}\|_{\infty,D_{\delta,\alpha}} + \left\| \frac{f+\widetilde{f}}{2} \right\|_{\infty,D_{\delta,\alpha}}^2 &\le 2\varepsilon' + \frac{1}{4}(1+\varepsilon')^2 \\
	&<1.
	\label{equation:robustness-condition-application}
\end{align}
Thus, the condition of Lemma \ref{lemma:robustness of block-encoding} is satisfied. Therefore, there exists a block encoding $U_f$ of $f$ such that
\begin{align}
	\|U_f-U_{\widetilde{f}}\|_\infty &\le \frac{\varepsilon'}{2}\sqrt{\frac{2}{1- \frac{1}{4}}} \\
    &=\sqrt{\frac{2}{3}}\varepsilon'\\
	&< \frac{1}{120}.
\end{align}
Consequently, using the relationship between the diamond norm and the operator norm \cite{Kitaev1997}, we obtain
\begin{equation}
	\|\mathcal{U}_f-\mathcal{U}_{\widetilde{f}}\|_\diamond \le 2\|U_f-U_{\widetilde{f}}\|_\infty < \frac{1}{60}.
\end{equation}
If the approximate channels $\mathcal{E}_\rho$ and $\mathcal{E}_\rho^\dagger$ are used $l$ times in total, a hybrid argument gives
\begin{equation}
	\|\mathcal{E}_{\widetilde{f}}-\mathcal{U}_{\widetilde{f}}\|_\diamond \le l\varepsilon.
\end{equation}
Overall,
\begin{equation}
	\|\mathcal{E}_{\widetilde{f}}-\mathcal{U}_f\|_\diamond \le l\varepsilon + \frac{1}{60}=:\tau.
\end{equation}

Now, let us focus on the discrimination task using the above QSVT construction.

\begin{Lem}
	Suppose that $\rho$ is promised to be either $\rho_+$ or $\rho_-$ in \eqref{equation:discriminated quantum state}, and that we are given channels $\mathcal{E}_\rho$ and $\mathcal{E}_\rho^\dagger$ that approximate $\mathcal{U}_\rho$ and $\mathcal{U}_\rho^\dagger$ to diamond-norm precision $\varepsilon$.
	Choose the parameter $\delta$ defining $\rho_+$ and $\rho_-$ so that $\delta=\Theta(\varepsilon)$.
	Then there exists an algorithm that identifies whether the given channels correspond to $\rho_+$ or $\rho_-$ with constant success probability, using
	\begin{equation}
		\bigO{\frac{1}{\varepsilon}}
	\end{equation}
	queries to $\mathcal{E}_\rho$ and $\mathcal{E}_\rho^\dagger$.
	\label{lemma:discriminate with queries}
\end{Lem}

\begin{proof}
	Suppose that we have access to $\mathcal{E}_\rho$ and $\mathcal{E}_\rho^\dagger$.
	Using the QSVT circuit from Lemma~\ref{lemma:QSVT}, we obtain the channel $\mathcal{E}_{\widetilde{f}}$.
	The two channels $\mathcal{E}_\rho$ and $\mathcal{E}_\rho^\dagger$ are used the same number of times.
	We apply $\mathcal{E}_{\widetilde{f}}$ to $\ket{0}\!\bra{0}^{\otimes(a+2)}$ and measure the projector
	\begin{equation}
		\Pi = \ket{0}\!\bra{0}^{\otimes(a+2)}.
	\end{equation}
	For every state $\sigma$,
	\begin{equation}
		\left| \Tr\left[ \Pi\mathcal{E}_{\widetilde{f}}(\sigma) \right] - \Tr\left[ \Pi\mathcal{U}_f(\sigma) \right] \right| \le \|\mathcal{E}_{\widetilde{f}}-\mathcal{U}_f\|_\diamond \le l\varepsilon+\frac{1}{60}.
	\end{equation}

    If $\rho=\rho_+$, then
    \begin{equation}
        f\left(\frac{1/2+\delta}{\alpha}\right)=0.
    \end{equation}
    Choose $\delta=C\varepsilon$, where $C>0$ is a sufficiently large constant such that $l\varepsilon\le 1/120$.
    Such a constant exists because $\alpha=\Theta(1)$ and $\varepsilon'=0.01$ is fixed.
    Therefore, in each repetition,
    \begin{equation}
        p_+:=\Pr[\widehat{\Pi}]\le\tau\le\frac{1}{40}.
    \end{equation}

	On the other hand, if $\rho=\rho_-$, then
    \begin{equation}
        f\left(\frac{1/2-\delta}{\alpha}\right)=\frac{1}{2}.
    \end{equation}
    Since the corresponding ideal postselection probability is $1/4$,
    \begin{equation}
        p_-:=\Pr[\widehat{\Pi}]\ge\frac{1}{4}-\tau\ge\frac{9}{40}.
    \end{equation}
    Repeat the above procedure independently five times.
    Output $\rho_-$ if the outcome associated with $\widehat{\Pi}$ occurs at least once; otherwise output $\rho_+$.
    If $\rho=\rho_+$, the success probability is at least
    \begin{equation}
        (1-p_+)^5\ge\left(\frac{39}{40}\right)^5>\frac{2}{3}.
    \end{equation}
    If $\rho=\rho_-$, the success probability is at least
    \begin{equation}
        1-(1-p_-)^5\ge1-\left(\frac{31}{40}\right)^5>\frac{2}{3}.
\end{equation}
	Thus, we can discriminate the two cases with constant success probability.
	The query complexity is
	\begin{equation}
		\bigO{\frac{1}{\delta}} = \bigO{\frac{1}{\varepsilon}}.
	\end{equation}
\end{proof}

\begin{proof}[Proof of Theorem~\ref{Theorem:lower bound of samples}]
	Suppose that there exists a protocol which, given $N$ copies of $\rho$, implements channels $\mathcal{E}_\rho$ and $\mathcal{E}_\rho^\dagger$ that are $\varepsilon$-close in diamond norm to $\mathcal{U}_\rho$ and $\mathcal{U}_\rho^\dagger$, respectively.
	Here, $\rho$ is promised to be either $\rho_+$ or $\rho_-$.
	By Lemma~\ref{lemma:discriminate with queries}, one can discriminate the two cases using $\bigO{1/\varepsilon}$ queries to $\mathcal{E}_\rho$ and $\mathcal{E}_\rho^\dagger$.
	Therefore, the resulting discrimination procedure uses
	\begin{equation}
		\bigO{\frac{N}{\varepsilon}}
	\end{equation}
	copies of $\rho$.
	By Lemma~\ref{Lemma:lower bound of sample to discriminate the states}, however, every such discrimination procedure requires $\bigOmega{1/\delta^2}=\bigOmega{1/\varepsilon^2}$ copies since $\delta=\Theta(\varepsilon)$.
	Consequently,
	\begin{equation}
		N = \bigOmega{\frac{1}{\varepsilon}}.
	\end{equation}
\end{proof}

\section{Proof of Theorem~\ref{theorem:lower-bound-of-queries}}
\label{appendix:query-to-sample}

Fix positive numbers $\lambda_1,\cdots,\lambda_r$ satisfying
\begin{equation}
    \sum_{j=1}^r\lambda_j=1,
\end{equation}
and denote
\begin{equation}
    \lambda_{\mathrm{max}}:=\max_{j\in[r]}\lambda_j
\end{equation}
We assume that $\lambda_\mathrm{max}\le\alpha$.
For every subset $S\subset[d]$ of cardinality $r$, write
\begin{equation}
    S=\{s_1<\cdots<s_r\},
\end{equation}
and define $f:[d]\to[0,\pi/2]$ by
\begin{equation}
    f(x):=
    \begin{cases}
        \arcsin\left(\dfrac{\lambda_j}{\alpha}\right),
        & x=s_j,\quad j\in[r],\\
        0,
        & x\notin S.
    \end{cases}
\end{equation}

We define the phase oracle $O_f$ by
\begin{equation}
	O_f\ket{x} = e^{if(x)}\ket{x}.
\end{equation}
Here, $\{\ket{x}\}_{x\in[d]}$ denotes the computational basis.
The dependence of $f$ and $O_f$ on $S$ is implicit.
By construction, the quantity
$\lambda_{\mathrm{max}}=\max_{j\in[r]}\lambda_j$
is independent of $S$.

Since $0\le\lambda_j\le\alpha$ for every $j\in[r]$, the quantity $\arcsin(\lambda_j/\alpha)$ is well-defined.
Following \cite{Grover1996, Bennett1997, Zalka1999, Boyer1999}, we prove a lower bound for the corresponding quantum search problem.

\begin{Lem}
	Let $\eta\in(0,1/2)$.
    Let $\eta\in(0,1/2)$.
    Suppose that an algorithm, given query access to the controlled phase oracle $\mathrm{c}\text{-}O_f$ and its inverse $\mathrm{c}\text{-}O_f^\dagger$ corresponding to an unknown subset $S\subset[d]$ of cardinality $r$, is required to output $\hat{z}\in S$ with success probability at least $1-\eta$ for every such subset $S$. Then the algorithm requires
	\begin{equation}
		\Omega\!\left( \frac{\alpha}{\lambda_\mathrm{max}} \sqrt{\frac{d}{r}} \left(\frac{1}{2}-\eta\right) \right)
	\end{equation}
	queries to $\mathrm{c}\text{-}O_f$ and $\mathrm{c}\text{-}O_f^\dagger$.
	\label{Lem:BVVV}
\end{Lem}

\begin{proof}
	Write the controlled phase oracle and its inverse as
	\begin{equation}
		Q_f^\pm := \ket{0}\!\bra{0}\otimes I + \ket{1}\!\bra{1} \otimes \left[ I + \sum_{j=1}^r \left( e^{\pm i\arcsin(\lambda_j/\alpha)} -1 \right) \ket{s_j}\!\bra{s_j} \right].
	\end{equation}
	The fixed oracle is
	\begin{equation}
		Q_0 := \ket{0}\!\bra{0}\otimes I + \ket{1}\!\bra{1}\otimes I = I.
	\end{equation}

	For an initial state $\ket{\psi_{\mathrm{init}}}$, consider the two final states after $T$ queries:
	\begin{align}
		\ket{\psi_T^{(f)}} &= U_TQ_f^{y_{T-1}}U_{T-1} Q_f^{y_{T-2}} \cdots Q_f^{y_0}U_0 \ket{\psi_{\mathrm{init}}}, \\
		\ket{\psi_T^{(0)}} &= U_TQ_0U_{T-1} Q_0 \cdots Q_0U_0 \ket{\psi_{\mathrm{init}}},
	\end{align}
	where $y_t\in\{+,-\}$, and each $U_t$ is an arbitrary unitary independent of the oracle.
	Let $\Pr_f(\hat{z}\in S)$ and $\Pr_0(\hat{z}\in S)$ be the probabilities of outputting an element of $S$ after measuring the corresponding final state.
	For the POVM effect $E_S$ associated with this event,
	\begin{equation}
		\left| \Tr\left[ E_S \left( \ket{\psi_T^{(f)}}\!\bra{\psi_T^{(f)}} - \ket{\psi_T^{(0)}}\!\bra{\psi_T^{(0)}} \right) \right] \right| \le \frac{1}{2} \left\| \ket{\psi_T^{(f)}}\!\bra{\psi_T^{(f)}} - \ket{\psi_T^{(0)}}\!\bra{\psi_T^{(0)}} \right\|_1.
	\end{equation}
	Thus,
	\begin{equation}
		\Pr_f(\hat{z}\in S) - \Pr_0(\hat{z}\in S) \le \frac{1}{2} \left\| \ket{\psi_T^{(f)}}\!\bra{\psi_T^{(f)}} - \ket{\psi_T^{(0)}}\!\bra{\psi_T^{(0)}} \right\|_1.
	\end{equation}
	Using
	\begin{equation}
		\left\| \ket{\psi_T^{(f)}}\!\bra{\psi_T^{(f)}} - \ket{\psi_T^{(0)}}\!\bra{\psi_T^{(0)}} \right\|_1 \le 2 \left\| \ket{\psi_T^{(f)}} - \ket{\psi_T^{(0)}} \right\|_2,
	\end{equation}
	define
	\begin{equation}
		\delta_S(t) := \left\| \ket{\psi_t^{(f)}} - \ket{\psi_t^{(0)}} \right\|_2.
	\end{equation}
	We obtain the necessary condition
	\begin{equation}
		\Pr_f(\hat{z}\in S) - \Pr_0(\hat{z}\in S) \le \delta_S(T).
		\label{equation:necessary condition}
	\end{equation}

	Let
	\begin{equation}
		\Pi_S := \sum_{x\in S}\ket{x}\!\bra{x}.
	\end{equation}
	For every state $\ket{\psi}$,
	\begin{align}
		\left\| \left(Q_f^\pm-Q_0\right)\ket{\psi} \right\|_2 
        &= \left\| \left[ \ket{1}\!\bra{1} \otimes \sum_{j=1}^r \left( e^{\pm i\arcsin(\lambda_j/\alpha)} -1 \right) \ket{s_j}\!\bra{s_j} \right] \ket{\psi} \right\|_2 \\
		&\le \left| e^{\pm i\arcsin(\lambda_\mathrm{max}/\alpha)} -1 \right| \left\| (I\otimes\Pi_S)\ket{\psi} \right\|_2 \\
		&= 2 \left| \sin\left( \frac{\arcsin(\lambda_\mathrm{max}/\alpha)}{2} \right) \right| \left\| (I\otimes\Pi_S)\ket{\psi} \right\|_2 \\
		&\le 2 \frac{\lambda_\mathrm{max}}{\alpha} \left\| (I\otimes\Pi_S)\ket{\psi} \right\|_2.
	\end{align}
	The last inequality follows from
	\begin{equation}
		0 \le \sin\left(\frac{\arcsin y}{2}\right) \le y \qquad (0\le y\le1).
	\end{equation}

	Using the preceding bound,
	\begin{align}
		\delta_S(t+1) &= \left\| U_{t+1} \left( Q_f^{y_t}\ket{\psi_t^{(f)}} - Q_0\ket{\psi_t^{(0)}} \right) \right\|_2 \\
		&= \left\| Q_f^{y_t}\ket{\psi_t^{(f)}} - Q_0\ket{\psi_t^{(0)}} \right\|_2 \\
		&\le \left\| Q_f^{y_t} \left( \ket{\psi_t^{(f)}}-\ket{\psi_t^{(0)}} \right) \right\|_2 + \left\| \left( Q_f^{y_t}-Q_0 \right) \ket{\psi_t^{(0)}} \right\|_2 \\
		&\le \delta_S(t) + 2 \frac{\lambda_\mathrm{max}}{\alpha} \left\| (I\otimes\Pi_S) \ket{\psi_t^{(0)}} \right\|_2.
	\end{align}
	Since $\delta_S(0)=0$, iterating gives
	\begin{align}
		\delta_S(T) &\le 2 \frac{\lambda_\mathrm{max}}{\alpha} \sum_{t=0}^{T-1} \left\| (I\otimes\Pi_S) \ket{\psi_t^{(0)}} \right\|_2 \notag \\
		&= 2 \frac{\lambda_\mathrm{max}}{\alpha} \sum_{t=0}^{T-1} \sqrt{ \sum_{x\in S}p_t(x) },
		\label{equation:hybrid-distance-bound}
	\end{align}
	where $p_t(x)$ is the probability of obtaining $x$ upon measuring the second register of $\ket{\psi_t^{(0)}}$ with the POVM
	\begin{equation}
		\left\{ I\otimes\ket{x}\!\bra{x} \right\}_{x\in[d]}.
	\end{equation}

	Averaging \eqref{equation:necessary condition} over a uniformly random subset $S\subset[d]$ of size $r$, we obtain 
    \begin{equation}
        \mathbb{E}_S[\delta_S(T)] \ge
        \mathbb{E}_S[\Pr_f(\hat{z}\in S)] - \mathbb{E}_S[\Pr_0(\hat{z}\in S)].
    \end{equation}
	Since the fixed-oracle output is independent of $S$,
	\begin{equation}
		\mathbb{E}_S \left[ \Pr_0(\hat{z}\in S) \right] = \frac{r}{d}.
	\end{equation}
    Moreover, the success condition implies $\mathbb{E}_S[\Pr_f(\hat{z}\in S)] \ge 1-\eta$.
    Therefore,
	\begin{equation}
		\mathbb{E}_S[\delta_S(T)] \ge 1-\eta-\frac{r}{d}.
		\label{equation:hybrid-distance-lower-bound}
	\end{equation}
	On the other hand, Jensen's inequality and \eqref{equation:hybrid-distance-bound} imply
	\begin{align}
		\mathbb{E}_S[\delta_S(T)] &\le 2 \frac{\lambda_\mathrm{max}}{\alpha} \sum_{t=0}^{T-1} \mathbb{E}_S \left[ \sqrt{\sum_{x\in S}p_t(x)} \right] \\
		&\le 2 \frac{\lambda_\mathrm{max}}{\alpha} \sum_{t=0}^{T-1} \sqrt{ \mathbb{E}_S \left[ \sum_{x\in S}p_t(x) \right] } \\
		&= 2 \frac{\lambda_\mathrm{max}}{\alpha} \sum_{t=0}^{T-1} \sqrt{ \frac{r}{d} \sum_xp_t(x) } \\
		&= 2 \frac{\lambda_\mathrm{max}}{\alpha} T \sqrt{\frac{r}{d}}.
		\label{equation:hybrid-distance-upper-bound}
	\end{align}
	Combining \eqref{equation:hybrid-distance-lower-bound} and \eqref{equation:hybrid-distance-upper-bound},
	\begin{equation}
		T \ge \frac{\alpha}{ 2\lambda_\mathrm{max} } \sqrt{\frac{d}{r}} \left( 1-\eta-\frac{r}{d} \right).
	\end{equation}
	Since $r\le d/2$, this proves the claim.
\end{proof}

\begin{Lem}
	\label{lem:block-encoding}
	Let $a\in\mathbb{Z}^+$, let $S=\{s_1<\cdots<s_r\}\subset[d]$. Suppose that $\alpha\ge\lambda_\mathrm{max}$.
	Define $U_f$ by the circuit in Fig.~\ref{fig:circuit1}.
	Then $U_f$ is an $(\alpha,a,0)$-block encoding of
	\begin{equation}
		\rho_S:= \sum_{j=1}^r \lambda_j\ket{s_j}\!\bra{s_j}_X,
	\end{equation}
	namely,
	\begin{equation}
		\left( \bra{0}_A^{\otimes a}\otimes I_X \right) U_f \left( \ket{0}_A^{\otimes a}\otimes I_X \right) = \sum_{j=1}^r \frac{\lambda_j}{\alpha} \ket{s_j}\!\bra{s_j}_X.
	\end{equation}
\end{Lem}

\begin{proof}
	Starting from $\ket{0}_A^{\otimes a}\ket{x}_X$, we trace the action of the circuit in Fig.~\ref{fig:circuit1}:
	\begin{align}
		\ket{0}_A^{\otimes a-1}\ket{0}_A\ket{x}_X &\xrightarrow{H} \ket{0}_A^{\otimes a-1} \frac{ \ket{0}_A+\ket{1}_A }{\sqrt{2}} \ket{x}_X \\
		&\xrightarrow{\mathrm{c}\text{-}O_f} \ket{0}_A^{\otimes a-1} \frac{ \ket{0}_A+e^{if(x)}\ket{1}_A }{\sqrt{2}} \ket{x}_X \\
		&\xrightarrow{X} \ket{0}_A^{\otimes a-1} \frac{ \ket{1}_A+e^{if(x)}\ket{0}_A }{\sqrt{2}} \ket{x}_X \\
		&\xrightarrow{\mathrm{c}\text{-}O_f^\dagger} \ket{0}_A^{\otimes a-1} \frac{ e^{-if(x)}\ket{1}_A + e^{if(x)}\ket{0}_A }{\sqrt{2}} \ket{x}_X \\
		&\xrightarrow{H} \ket{0}_A^{\otimes a-1} \frac{ \left(e^{if(x)}+e^{-if(x)}\right)\ket{0}_A + \left(e^{if(x)}-e^{-if(x)}\right)\ket{1}_A }{2} \ket{x}_X \\
		&\xrightarrow{X} \ket{0}_A^{\otimes a-1} \frac{ \left(e^{if(x)}-e^{-if(x)}\right)\ket{0}_A + \left(e^{if(x)}+e^{-if(x)}\right)\ket{1}_A }{2} \ket{x}_X \\
		&\xrightarrow{e^{-i\pi/2}} \ket{0}_A^{\otimes a-1} \left( \sin f(x)\ket{0}_A\ket{x}_X - i\cos f(x)\ket{1}_A\ket{x}_X \right).
	\end{align}
	Therefore,
	\begin{equation}
		U_f \left( \ket{0}_A^{\otimes a}\ket{x}_X \right) = \ket{0}_A^{\otimes a-1} \left( \sin f(x)\ket{0}_A\ket{x}_X - i\cos f(x)\ket{1}_A\ket{x}_X \right).
	\end{equation}
	The amplitude on $\ket{0}_A^{\otimes a}$ is $\sin f(x)$, which equals $\lambda_j/\alpha$ when $x=s_j$ and vanishes when $x\notin S$.
	Hence,
	\begin{equation}
		\left( \bra{0}_A^{\otimes a}\otimes I_X \right) U_f \left( \ket{0}_A^{\otimes a}\otimes I_X \right) = \sum_{j=1}^r \frac{\lambda_j}{\alpha} \ket{s_j}\!\bra{s_j}.
	\end{equation}
	Thus, $U_f$ is an $(\alpha,a,0)$-block encoding of $\rho_S$.
\end{proof}

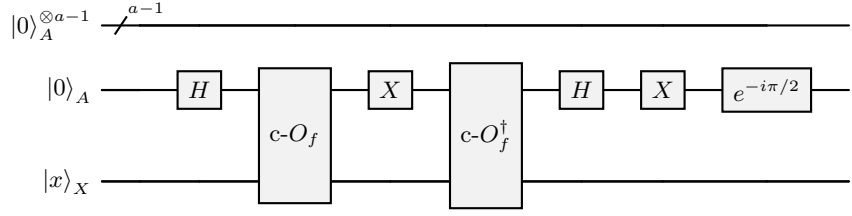
\begin{figure*}[t]
	\centering
	\begin{quantikz}[background color=black!5!white]
		\lstick{$\ket{0}_A^{\otimes a-1}$} & \qwbundle{a-1} & \qw & \qw & \qw & \qw & \qw & \qw & \qw &  \\
		\lstick{$\ket{0}_A$} & \qw & \gate{H} & \gate[2]{\mathrm{c}\text{-}O_f} & \gate{X} & \gate[2]{\mathrm{c}\text{-}O_f^\dagger} & \gate{H} & \gate{X} & \gate{e^{-i\pi/2}} &  \\
		\lstick{$\ket{x}_X$} & \qw & \qw & \qw & \qw & \qw & \qw & \qw & \qw & \qw
	\end{quantikz}
	\caption{The quantum circuit implementing the block-encoding unitary channel $\mathcal{U}_f$ induced by $U_f$.}
	\label{fig:circuit1}
\end{figure*}

\begin{proof}[Proof of Theorem~\ref{theorem:lower-bound-of-queries}]
	For every subset $S=\{s_1<\cdots<s_r\}\subset[d]$, define
	\begin{equation}
		\rho_S:=\sum_{j=1}^r\lambda_j\ket{s_j}\!\bra{s_j}.
	\end{equation}
	Since a state-recovery algorithm must succeed for every unknown
	rank-$r$ state, it must in particular succeed for every state
	$\rho_S$ in this family.
	Assume that there exists a recovery algorithm $\mathcal{A}$ that, for every $S$, uses at most $N$ queries in total to the block-encoding unitary channel $\mathcal{U}_f$ and $\mathcal{U}_f^\dagger$, induced by an $(\alpha,a,0)$-block-encoding unitary $U_f$ of $\rho_S$ and its inverse, outputs a state $\hat{\rho}_S$ satisfying
	\begin{equation}
		\|\hat{\rho}_S-\rho_S\|_1 \le \varepsilon.
	\end{equation}
    Here, $N$ denotes the worst-case number of queries over $S$.
    
	Using this subroutine, we construct a search algorithm with access to $\mathrm{c}\text{-}O_f$ and $\mathrm{c}\text{-}O_f^\dagger$.
	First, we implement $\mathcal{U}_f$ from Lemma~\ref{lem:block-encoding} using a constant number of queries to these controlled phase oracles.
	Since $\sum_{j=1}^r\lambda_j=1$, the unitary $U_f$ is an $(\alpha,a,0)$-block encoding of the quantum state
	\begin{equation}
		\rho_S = \sum_{j=1}^r \lambda_j\ket{s_j}\!\bra{s_j}.
	\end{equation}
	We then run $\mathcal{A}$ with access to $\mathcal{U}_f$ and $\mathcal{U}_f^\dagger$ to obtain $\hat{\rho}_S$, and measure it in the computational basis.

	Let
	\begin{equation}
		p_S(x) := \bra{x}\rho_S\ket{x}, \qquad q_S(x) := \bra{x}\hat{\rho}_S\ket{x}
	\end{equation}
	be the distributions induced by this measurement.
	By data processing under the computational-basis measurement,
	\begin{equation}
		d_{\mathrm{TV}}(p_S,q_S) = \frac{1}{2} \sum_x|p_S(x)-q_S(x)| \le \frac{1}{2} \|\hat{\rho}_S-\rho_S\|_1 \le \frac{\varepsilon}{2}.
	\end{equation}
	Since $\sum_{z\in S}p_S(z)=1$, the variational characterization of total variation distance gives
	\begin{equation}
		1-\sum_{z\in S}q_S(z) = \sum_{z\in S}p_S(z)-\sum_{z\in S}q_S(z) \le d_{\mathrm{TV}}(p_S,q_S) \le \frac{\varepsilon}{2}.
	\end{equation}
	Therefore,
	\begin{equation}
		\sum_{z\in S}q_S(z) \ge 1-\frac{\varepsilon}{2}.
	\end{equation}

	By Lemma~\ref{lem:block-encoding}, each query to $\mathcal{U}_f$ or $\mathcal{U}_f^\dagger$ uses only a constant number of queries to $\mathrm{c}\text{-}O_f$ and $\mathrm{c}\text{-}O_f^\dagger$.
	Thus, the constructed search algorithm uses $O(N)$ phase-oracle queries in the worst case over $S$ and succeeds with probability at least $1-\varepsilon/2$ for every $S$.
	Since every state in the family satisfies
	\begin{equation}
		\lambda_{\max}(\rho_S)=\lambda_{\mathrm{max}},
	\end{equation}
	the same bound can be written in terms of the largest eigenvalue of the state.
	Denoting a worst-case member of this family by $\rho$, we obtain
	\begin{equation}
		N=\bigOmega{\frac{\alpha}{\lambda_{\max}(\rho)}\sqrt{\frac{d}{r}}\left(1-\frac{\varepsilon}{2}-\frac{r}{d}\right)}.
	\end{equation}
	Since $r/d\le1/2$,
	\begin{equation}
		N = \bigOmega{ \frac{\alpha}{\lambda_{\max}(\rho)} \sqrt{\frac{d}{r}} (1-\varepsilon) }.
	\end{equation}
\end{proof}
\end{document}